\documentclass[12pt,english]{article}
\usepackage{lmodern}
\usepackage[T1]{fontenc}
\usepackage[utf8]{inputenc}
\usepackage[a4paper,margin=2.5cm]{geometry}
\usepackage{color}
\usepackage{babel}
\usepackage{float}
\usepackage{mathtools}
\usepackage{amsmath}
\usepackage{amssymb}
\usepackage{comment}
\usepackage{appendix}
\usepackage{enumitem}
\usepackage[authoryear]{natbib}
\usepackage[unicode=true,
 bookmarks=false,
 breaklinks=false,pdfborder={0 0 1},
 colorlinks=true]
 {hyperref}
\hypersetup{
 urlcolor=refcolor,citecolor=refcolor}

\providecommand{\keywords}[1]{\textbf{Keywords:} #1}
\providecommand{\JEL}[1]{\textbf{JEL Classification:} #1}

\makeatletter

\usepackage{amsfonts}
\usepackage{amsthm}
\usepackage{tablefootnote}
\usepackage{xcolor}
\usepackage{url}
\usepackage{tabularx}
\usepackage{graphicx}
\usepackage{caption}
\usepackage{tikz}
\usepackage{pgfplots}
\pgfplotsset{compat=newest}
\usepgfplotslibrary{fillbetween}

\definecolor{refcolor}{rgb}{0, 0, 0.5}

\newcommand{\Rho}{\mathrm{P}}

\newtheorem{proposition}{Proposition}

\newtheorem{lemma}{Lemma}
\theoremstyle{plain}

\makeatother

\title{\date{\today}} 

\title{Changing Simplistic Worldviews\thanks{We would like to thank Ole Jann, Artyom Jelnov, Pavel Kocourek, Fedor Sandomirskiy, and Jan Zapal for their valuable comments and suggestions.}}
\author{
        Maxim Senkov\thanks{European Research University, U Haldy 200/18, 700 30 Ostrava, Czech Republic. email: maxim.senkov@eruni.org.} \\
        \and
	Toygar T. Kerman\thanks{Institute of Economics, Corvinus University of Budapest, F\H{o}v\'{a}m t\'{e}r 8, 1093 Budapest, Hungary. email: toygar.kerman@uni-corvinus.hu. ORCID: 0000-0003-3038-3666. The author acknowledges funding by the Hungarian National Research, Development and Innovation Office, Project Number K-143276.}  
}

\begin{document}

\maketitle

\begin{abstract}
\noindent We study a Bayesian persuasion model with two-dimensional states of the world, in which the sender (she) and receiver (he) have heterogeneous prior beliefs and care about different dimensions. 
The receiver is a naive agent who has a \emph{simplistic worldview}: he ignores the dependency between the two dimensions of the state. 
We provide a characterization for the sender's gain from persuasion both when the receiver is naive and when he is rational.
We show that the receiver benefits from having a simplistic worldview if and only if it makes him perceive the states in which his interest is aligned with the sender as less likely.
\end{abstract}

\noindent\keywords{Bayesian persuasion; misspecified prior; correlation neglect}

\noindent\JEL{D82, D83, D91}

\section{Introduction}
Standard models of Bayesian persuasion assume that the sender (she) and receiver(s) (he) care about the same dimension of the true state of the world. 
For example, in their seminal paper \cite{kamenica2011bayesian} consider a prosecutor (sender) and a judge (receiver), both of whom are only interested in one dimension of the true state, that is whether the defendant is guilty or innocent. 
What if the sender and receiver are concerned about different dimensions of the true state? 

Consider a company that wishes to build a new factory and only cares about the project's profitability. 
In order to be perceived as environmentally friendly, the company asks a sustainability consultant to evaluate the project and provide a recommendation on whether to continue or terminate it.
However, the consultant only cares about a different aspect of the project: its sustainability.  
How should the consultant conduct her research to convince the company to continue the project only when it is sustainable? 

When considering multidimensional states of the world, however, a behavioral aspect emerges: thinking about how multiple issues relate to each other is cognitively demanding. 
Experimental evidence suggests that people find it challenging to work with joint distributions of random variables and tend to underestimate or ignore the correlation between state variables or observed signals \citep{kallir2009neglect,eyster2011correlation,enke2019correlation}.

In this paper, we study how the receiver's \emph{simplistic worldview} (i.e. treating dependent state variables as independent) affects the sender's ability to influence his decision. 
We extend \cite{kamenica2011bayesian} by $(i)$ allowing for two-dimensional states of the world and $(ii)$ considering a receiver who ignores the dependency of state variables (which leads to heterogeneous priors). 
We model the receiver's disregard of state dependency as a misspecified prior belief; while the receiver knows the correct marginal distribution for each dimension, he assumes that they are independent. 
We characterize optimal disclosure both when the receiver is \emph{rational} (i.e. does not ignore dependency) and \emph{naive} (i.e. assumes that dimensions are independent), and analyze the welfare effects of the receiver's simplistic worldview.

\subsection{Illustrative example}\label{sec:ex}
Let us return to our example of a company and sustainability consultant. 
Both the company and the consultant are initially uncertain whether the project is \emph{profitable} $(P)$ or \emph{loss-making} $(L)$ and \emph{sustainable} $(S)$ or \emph{unsustainable} $(U)$.
The company chooses to either \emph{continue} or \emph{terminate} the project and only cares about its profitability. 
Continuing the project when it is profitable and terminating it when it is loss-making provide the company a utility of 1, otherwise, the company's utility is 0. 
On the other hand, the consultant only cares about the sustainability of the project.
Continuing the project when it is sustainable and terminating it when it is unsustainable provide the consultant a utility of 1, otherwise the consultant's utility is 0. 

Hence, the true state of the world has two dimensions, and the conflict of interest depends on the joint distribution of the dimensions. 
First, suppose that the company is rational.
That is, it shares a common prior with the consultant, where the prior distribution of the states is given as follows (with marginal distributions on the sides). 
\begin{table}[H]
\begin{centering}
\begin{tabular}{|c|c|c|c}
\cline{1-3} \cline{2-3} \cline{3-3} 
 & $L$ & $P$ & \tabularnewline
\cline{1-3} \cline{2-3} \cline{3-3} 
$U$ & 0.25 & 0.1 & 0.35\tabularnewline
\cline{1-3} \cline{2-3} \cline{3-3} 
$S$ & 0.35 & 0.3 & 0.65 \tabularnewline
\cline{1-3} \cline{2-3} \cline{3-3} 
\multicolumn{1}{c}{} & \multicolumn{1}{c}{0.6} & \multicolumn{1}{c}{0.4} & \tabularnewline
\end{tabular}
\par\end{centering}
\caption{Prior belief distribution.}
\label{tab:ex}
\end{table}

\noindent The company asks the consultant to investigate the profitability and sustainability of the project and provide a recommendation to either \emph{continue} $(c)$ or \emph{terminate} $(t)$ it. 
Suppose that if the company is indifferent between continuing and terminating the project, it continues. 
Since the project is initially more likely to be loss-making, the company would (ex ante) prefer the project to be terminated (which yields expected payoff $(0.65)\cdot 0+(0.35)\cdot 1=0.35$ to the consultant and $(0.6)\cdot 1+(0.4)\cdot 0=0.6$ to the company).
However, the consultant can improve her expected payoff by strategically disclosing information about the true state of the world. 
Mathematically, the information disclosure can be represented by a \emph{signal} $\pi$ that sends recommendations contingent on the state of the world. 
Consider an optimal signal given as follows.
\begin{table}[H]
\centering
\begin{tabular}{c|cccc}
$ \pi $ & $ (U,L) $ & $ (U,P) $ & $(S,L)$ & $(S,P)$ \\
\hline
$c$ & 0 & 0.5 & 1 & 1 \\[0.1cm]
$t$ & 1 & 0.5 & 0 & 0 
\end{tabular}
\end{table}

\noindent First, observe that $\pi$ recommends continuing with probability 1 when the project is both sustainable and profitable (i.e. in  state $(S,P)$) and recommends terminating with probability 1 when the project is both unsustainable and loss-making (i.e. in  state $(U,L)$). 
This is intuitive since in these states the preferred actions of the sender and receiver are aligned. 
However, observe that $\pi$ recommends to continue the project with probability $1$ also when it is sustainable and loss-making.\footnote{Note that the optimal signal is not unique.}
The consultant chooses the probability of $c$ in state $(U,P)$ such that the company's posterior belief that the project is profitable is exactly $1/2$ and thus it chooses to continue. 
The consultant's expected payoff of employing $\pi$ is $(0.3)\cdot 1+(0.35)\cdot 1+0.1\cdot (0.5)+(0.25)\cdot 1=0.95$, improving upon the case of no disclosure. 
Similarly, the company's expected payoff from $\pi$ is $(0.3)\cdot 1+(0.35)\cdot 0+(0.1)\cdot (0.5)+(0.25)\cdot 1=0.6$, which is equal to the case of no disclosure.

Now suppose that the company is naive and ignores the dependency between the states of the world due to insufficient knowledge of sustainability. 
In other words, the company assumes that the dimensions of the true state are independent, and thus has a misspecified prior. 
While the company knows the marginal distributions $\Pr(S)=0.65$ and $\Pr(P)=0.4$, it assumes that $\Pr(S,P)=\Pr(S)\Pr(P)=(0.65)\cdot (0.4)=0.26$.\footnote{That is, the company assumes that $\Pr(P)=\Pr(P\vert S)$.} 
The company's \emph{simplistic worldview} regarding the prior distribution is given as follows. 
\begin{table}[H]
\begin{centering}
\begin{tabular}{|c|c|c|c}
\cline{1-3} \cline{2-3} \cline{3-3} 
 & $L$ & $P$ & \tabularnewline
\cline{1-3} \cline{2-3} \cline{3-3} 
$U$ & 0.21 & 0.14 & 0.35\tabularnewline
\cline{1-3} \cline{2-3} \cline{3-3} 
$S$ & 0.39 & 0.26 & 0.65 \tabularnewline
\cline{1-3} \cline{2-3} \cline{3-3} 
\multicolumn{1}{c}{} & \multicolumn{1}{c}{0.6} & \multicolumn{1}{c}{0.4} & \tabularnewline
\end{tabular}
\par\end{centering}
\caption{Misspecified prior belief distribution.}
\label{tab:ex2}
\end{table}

\noindent Notice that the company now perceives states $(S,L)$ and $(U,P)$ as relatively more likely than before. 
Therefore, the sender has two options to keep the recommendation persuasive: $(i)$ recommend to terminate more often than before when the project is loss-making (i.e. decrease probability of $c$ in $(S,L)$) or $(ii)$ recommend to continue more often than before when the project is profitable (i.e. increase probability of $c$ in $(U,P)$). 
In other words, the consultant needs to make the recommendation to continue \emph{more informative} regarding the project's profitability.
It turns out that given the misspecified prior of the receiver, it is optimal for the sender to follow the second route.  
\begin{table}[H]
\centering
\begin{tabular}{c|cccc}
$ \hat\pi $ & $ (U,L) $ & $ (U,P) $ & $(S,L)$ & $(S,P)$ \\
\hline
$c$ & 0 & 0.93 & 1 & 1 \\[0.1cm]
$t$ & 1 & 0.07 & 0 & 0 
\end{tabular}
\end{table}

\noindent Indeed, when the project is unsustainable and profitable, the recommendation to continue is sent with probability $0.93$, while it was sent with probability $0.5$ in this state before. 
The consultant's expected payoff of employing $\hat\pi$ is $(0.3)\cdot 1+(0.35)\cdot 1+(0.1)\cdot (0.07)+(0.25)\cdot 1=0.907$, still improving upon the case of no disclosure, albeit less than when the receiver is rational.  
On the other hand, the company's expected payoff from $\hat\pi$ is $(0.3)\cdot 1+(0.35)\cdot 0+(0.1)\cdot (0.93)+(0.25)\cdot 1=0.643$, which is higher relative to the case of having a correct prior belief (and also higher relative to the case of no disclosure).\footnote{We use misspecified model framework, which implies existence of correct prior belief (corresponding here to the sender's belief). Thus, the naive receiver's expected payoff is evaluated using the correct (sender's) prior belief.} 
Hence, while the company's simplistic worldview \emph{benefits} the company, it is \emph{detrimental} to the consultant. 

The example illustrates how the receiver \emph{benefits} from ignoring dependency between the state variables.
However, both $(i)$ the sender's approach regarding disclosure in states with misaligned preferences and $(ii)$ whether the receiver benefits from ignoring dependency crucially depend on the prior belief distribution.
In fact, the receiver might also be \emph{worse off} due to ignoring the dependency of states. 
In our analysis, we consider all possible prior belief distributions, provide a full characterization of the sender's gain from persuasion, and pin down the conditions under which the receiver benefits from having a simplistic worldview.  

\subsection{Related literature}
Our paper extends \cite{kamenica2011bayesian} in two aspects and contributes to three strands of literature. 
First, it relates to the literature on misspecified models. 
The receiver in our model knows the correct marginal distribution of each dimension, yet, assumes that the dimensions are independent.
Hence, the receiver possesses a misspecified prior belief about the states of the world. 
This naturally relates to studies that consider correlation neglect, as they motivate the behavioral assumption we make about the receiver. 
One main difference from such models (e.g. \cite*{levy2022persuasion}) is that our receiver ignores dependency between \emph{states}, whereas most papers in the literature focus on the receiver neglecting correlation between \emph{messages}.\footnote{\cite{eyster2005cursed} introduce \emph{cursed equilibrium} which assumes that players underestimate the correlation between other players' actions and information.} 

Our findings share similarities with those from studies focusing on the political applications of correlation neglect. 
\cite{ortoleva2015overconfidence} show that correlation neglect might lead to overconfidence in voters and that overconfidence is an important predictor for ideological extremeness. 
In our model, the misspecified prior belief of the receiver leads to the receiver having over-optimistic or over-pessimistic beliefs about the state of the world.  
\cite{levy2015correlation} show that voters neglecting the correlation between their information sources might contribute to higher political polarization in society, howbeit, voters might benefit from it.
Similarly, we illustrate that if the receiver ignores dependency between states then the optimal signal might be more informative.\footnote{This also relates to studies that investigate how a receiver can strategize to increase the informativeness of the signal \citep*{austen2000cheap,kartik2007note,ambrus2017delegation,tsakas2021resisting} and how different kinds of biases might affect persuasion \citep*{hagmann2017persuasion,augias2020persuading,de2022non}.} 

Our behavioral assumption regarding the receiver alters his perception of the joint belief distribution.
Yet, the receiver has no ambiguity about the message he observes. 
\cite*{eliaz2021persuasion} consider a sender who can redact his message once it is chosen, but the receiver only knows the message strategy and not the redaction strategy. 
In a related study, \cite*{eliaz2021strategic} consider a sender who can attach interpretations to messages. 
Unlike our model where the receiver's perception of the signal structure remains correct, both studies assume a distortion in the perception of the signal structure. 

Second, our paper contributes to the literature on Bayesian persuasion with heterogeneous priors.
\cite{alonso2016bayesian} assume that the sender and receiver have subjective prior beliefs and analyze its effect on the set of inducible posterior beliefs. 
\cite{galperti2019persuasion} focuses on a similar setup and considers receivers who might be reluctant to change their worldviews. 
In a related paper, \cite{kosterina2022persuasion} considers a sender who is ignorant about the receiver's prior and characterizes optimal information structures. 
In contrast to these studies, our model posits that the sender is aware of the accurate prior distribution while the receiver is not, and this discrepancy is common knowledge. 
Moreover, we contrast the receiver's payoff to the case in which he has the correct prior belief and and study the implications of incorrect prior for welfare. 

Third, our study relates to the literature on persuasion with multidimensional state variables. 
\cite{rayo2010optimal} consider a model of information disclosure and assume that the agent is rational, in which the sender's expected payoff is the product of posterior mean of the sender's and receiver's dimension of state, and characterize optimal disclosure. 
\cite{tamura2018bayesian} studies optimal disclosure in a more general setup similar to the model in \cite{rayo2010optimal}. 
In a related study, \cite{dworczak2019persuasion} characterize the optimal persuasion mechanism when the state is two-dimensional and the objective is quadratic. 
\cite{malamud2021persuasion} study a model where the sender observes the multidimensional state and show that it is optimal for the sender to ``reduce'' the dimensions of the state.
All studies above consider a rational sender and receiver. 
In contrast, our paper considers implications of the receiver's ignorance regarding the connection between the dimensions of state for the optimal signal and welfare.  
\cite*{babichenko2022regret} consider a sender who provides information to a receiver about a product, where each dimension represents a different attribute of the product. 
The authors extend the standard model also by assuming the sender does not know the utility of the receiver and wishes to minimize regret and show that ignorance of the utility of the receiver is extremely harmful to the sender. 
We, on the other hand, assume that agents have complete information about the utility functions, yet, the receiver  incorrectly assumes independence of the dimensions.

\section{Preliminaries}

\subsection{The model}
A sender (she) wishes to persuade a receiver (he) to take an action depending on the true state of the world. 
Let $\Omega=\Sigma\times R$ be the set of \emph{states} of the world, such that the sender's payoff depends only on $\Sigma=\{\sigma_0,\sigma_1\}$ and the receiver's payoff depends only on $R=\{\rho_0,\rho_1\}$.
For each $k,\ell\in\{0,1\}$, let $\omega_{k\ell}\equiv(\sigma_k,\rho_\ell)\in\Omega$.
 
A \emph{signal} consists of a finite set of \emph{messages} $M$ and a family of distributions $\{\pi(\cdot\vert\omega)\}_{\omega\in\Omega}$ over $M$.
That is, $\pi:\Omega\rightarrow\Delta(M)$ maps each state of the world to a probability distribution over $M$.
Denote the set of all signals by $\Pi$.

We start by considering a rational receiver, hence, the sender and receiver share a common prior belief $\mu\in\Delta(\Omega)$ about the true state of the world.
For each $k,\ell\in\{0,1\}$, let $\mu_{k\ell}\equiv\Pr(\sigma_k,\rho_\ell)$ and let the marginal distributions be denoted as $\mu(\sigma_k)\equiv\Pr(\sigma_k)$ and $\mu(\sigma_\ell)\equiv\Pr(\rho_\ell)$.
The distribution $\mu$ is presented in Table \ref{tab:prior} with marginal distributions given on the sides. 
\begin{table}[H]
\begin{centering}
\begin{tabular}{|c|c|c|c}
\cline{1-3} \cline{2-3} \cline{3-3} 
$(\Sigma,\Rho)$ & $\rho_0$ & $\rho_1$ & \tabularnewline
\cline{1-3} \cline{2-3} \cline{3-3} 
$\sigma_0$ & $\mu_{00}$ & $\mu_{01}$ & $\mu(\sigma_0)$\tabularnewline
\cline{1-3} \cline{2-3} \cline{3-3} 
$\sigma_1$ & $\mu_{10}$ & $\mu_{11}$ & \tabularnewline
\cline{1-3} \cline{2-3} \cline{3-3} 
\multicolumn{1}{c}{} & \multicolumn{1}{c}{$\mu(\rho_0)$} & \multicolumn{1}{c}{} & \tabularnewline
\end{tabular}
\par\end{centering}
\caption{Common prior $\mu$.}
\label{tab:prior}
\end{table}

Given $\mu\in\Delta(\Omega)$ and $\pi\in\Pi$, a message $m\in M^\pi$ generates the posterior belief $\mu^m\in\Delta(\Omega)$.
For each $k,\ell\in\{0,1\}$, the marginal posterior beliefs upon observing message $m$ are given by $\mu^m(\sigma_k)$ and $\mu^m(\rho_\ell)$, where 
\begin{equation*}
\mu^m(\rho_\ell)=\frac{\mu_{k\ell}\pi(m\vert \omega_{k\ell})+\mu_{\ell\ell}\pi(m\vert \omega_{\ell\ell})}{\mu_{k\ell}\pi(m\vert \omega_{k\ell})+\mu_{\ell\ell}\pi(m\vert \omega_{\ell\ell})+\mu_{\ell k}\pi(m\vert \omega_{\ell k})+\mu_{kk}\pi(m\vert \omega_{kk})}
\end{equation*}

\noindent and $\mu^m(\sigma_k)$ is given analogously.\footnote{We drop $\pi$ when denoting $\mu^m(\sigma_k)$ and $\mu^m(\rho_\ell)$ to ease notation.} 
We denote the posterior belief for the joint distribution upon observing message $m$ by $\mu^m(\sigma_k,\rho_\ell)$. 

The set of actions that are available to the receiver is given by $A=\{0,1\}$. 
The receiver wishes his action to match the dimension of the state he cares about, i.e., the receiver wants to choose 0 when the state is $(\sigma_k,\rho_0)$ and wants to choose 1 when the state is $(\sigma_k,\rho_1)$, for any $k\in\{0,1\}$. 
The receiver's utility function $u:A\times R\rightarrow\{0,1\}$ is given by 
\[
u\left(a,\rho\right)=\begin{cases}
1 & \text{if }\left(a,\rho\right)=\left(0,\rho_0\right)\text{ or }\left(a,\rho\right)=\left(1,\rho_1\right),\\
0 & \text{otherwise.}
\end{cases} 
\]

Given the receiver's preferences, it follows that the receiver is indifferent between $0$ and $1$ if $\mu^m(\rho_0)=\mu^m(\rho_1)=1/2$. 

The sender wishes the chosen action to match the dimension of the state she cares about. 
The sender's utility function $v:A\times\Sigma\rightarrow\{0,1\}$ is given by 
\[
v\left(a,\sigma\right)=\begin{cases}
1 & \text{if }\left(a,\sigma\right)=\left(0,\sigma_0\right)\text{ or }\left(a,\sigma\right)=\left(1,\sigma_1\right),\\
0 & \text{otherwise.}
\end{cases} 
\]

We assume that the the default action of the receiver is $0$, i.e. $\mu(\rho_1)<\mu(\rho_0)$.
Given $\pi\in\Pi$, the \emph{choice rule} $\alpha^\pi:M\rightarrow A$ of the receiver is given by
\begin{equation}
\alpha^\pi\left(m\right)=\begin{cases}
1 & \text{if }\mu^m(\rho_1)\geq\frac12,\\
0 & \text{otherwise.}
\end{cases}\label{action_strategy}
\end{equation}

\noindent That is, the receiver switches his action when he is indifferent upon observing a message.

\subsection{Benchmark}
We start by characterizing optimal disclosure when the receiver does not ignore the dependency of states, i.e., when he is  rational. 
This will serve as a benchmark for the next sections. 

First, note that since both the receiver's action set and the payoff-relevant state are binary, we can restrict attention to \emph{direct} signals without loss of generality.\footnote{A signal is \emph{direct} if it holds that ($i$) $M\subseteq A$ and ($ii$) for all $a\in M$, $\alpha^{\pi}(a)=a$.}
Therefore, we hereon assume that $M=\{0,1\}$ and that $\Pi$ is the set of all direct signals.  
We denote a signal by $\pi=(\pi_{00},\pi_{01},\pi_{10},\pi_{11})$ where $\pi_{k\ell}\equiv\pi(1\vert\omega_{k\ell})$, for $k,\ell\in\{0,1\}$.
 
Depending on the prior, the sender may choose $\pi\in\Pi$ to induce a switch of the receiver's default action $0$ to $1$.
For this to be successful, message $1$ should induce a posterior belief $\mu^{1}(\rho_1)\geq 1/2$.
We formalize the obedience constraint in the following lemma. 

\begin{lemma}\label{lem:str}
Given $\pi\in\Pi$, the receiver chooses action $1$ if and only if $$\mu_{00}\pi_{00}+\mu_{10}\pi_{10}\leq \mu_{01}\pi_{01}+\mu_{11}\pi_{11}.$$
\end{lemma}

Note that Lemma \ref{lem:str} implies that the obedience constraint for action $0$ given by $\mu_{01}(1-\pi_{01})+\mu_{11}(1-\pi_{11})<\mu_{00}(1-\pi_{00})+\mu_{10}(1-\pi_{10})$ is also satisfied. 
See Lemma \ref{lem:suff} in the appendix for the formal statement and proof. 

We can now state the sender's optimization problem, which is to maximize the ex-ante expected utility by choosing a signal:
\begin{align}
\max_{\pi\in\Pi}\: \mu_{11}\pi_{11}+\mu_{10}\pi_{10}&-\mu_{01}\pi_{01}-\mu_{00}\pi_{00}\nonumber\\
\text{s.t. } &\mu_{00}\pi_{00}+\mu_{10}\pi_{10}\leq \mu_{01}\pi_{01}+\mu_{11}\pi_{11}.\label{OP_soph}
\end{align}

Let $\pi^*$ denote an optimal signal when the receiver is rational.

\begin{proposition}\label{prop:bench}
An optimal signal when the receiver is rational is given by 
\begin{equation*}
(\pi^*_{00},\pi^*_{01},\pi^*_{10},\pi^*_{11})=\begin{cases}
(0,0,1,1) & \text{if }\mu_{11}\geq\mu_{10},\\
(0,\beta,\gamma,1) & \text{otherwise},
\end{cases}
\end{equation*}

\noindent where $\mu_{10}\gamma-\mu_{01}\beta=\mu_{11}$.
\end{proposition}

There are a few aspects to note about Proposition \ref{prop:bench}. 
First, it is intuitive that $\pi^*_{00}=0$ and $\pi^*_{11}=1$, since in states $(\sigma_0,\rho_0)$ and $(\sigma_1,\rho_1)$ the sender's and receiver's preferences are aligned.
In other words, it is optimal for the sender to \emph{disclose the truth} in these states, i.e., she sends message $1$ with probability 0 in state $(\sigma_0,\rho_0)$ and with probability 1 in state $(\sigma_1,\rho_1)$.
More precisely, when the state is $(\sigma_1,\rho_1)$ the sender has an incentive to \emph{increase} the frequency of message $1$, which switches the receiver's action. 
Similarly, when the state is $(\sigma_0,\rho_0)$ the sender has an incentive to \emph{decrease} the frequency of message $1$.
Note that both increasing $\pi_{11}$ and decreasing $\pi_{00}$ \emph{relaxes} the obedience constraint in (\ref{OP_soph}). 

Second, observe that while the sender has an incentive to \emph{decrease} $\pi_{01}$ in state $(\sigma_0,\rho_1)$ and \emph{increase} $\pi_{10}$ in state $(\sigma_1,\rho_0)$, this \emph{tightens} the obedience constraint since the sender's and receiver's preferences are misaligned. 
Therefore, the sender chooses $\pi_{01}$ and $\pi_{10}$ such that the obedience constraint binds. 
However, if the prior for $(\sigma_1,\rho_1)$ is (weakly) higher than the prior for $(\sigma_1,\rho_0)$ (i.e. $\mu_{11}\geq\mu_{10}$), then the effect of relaxing the constraint is stronger than the effect of tightening the constraint.
In this case, since $\pi^*_{01}=0$ and $\pi^*_{10}=1$, the sender can implement her desired outcome with certainty and obtain a payoff of 1. 

\section{Optimal signal when the receiver is naive}
From hereon we assume that the receiver is naive. 
More precisely, while the receiver knows $(\mu(\sigma_0),\mu(\sigma_1))$ and $(\mu(\rho_0),\mu(\rho_1))$, he perceives $\sigma$ and $\rho$ as independent random variables. 
This gives rise to the misspecified prior $\hat\mu$ about $\Omega$, where the receiver assumes that each element of the support of $\hat\mu$ is the product of corresponding marginals, i.e. $\hat\mu_{k\ell}=\mu^\sigma_k\mu^\rho_\ell$. 
We present the misspecified prior in the following table with marginal distributions written on the sides. 
\begin{table}[H]
\begin{centering}
\begin{tabular}{|c|c|c|c}
\cline{1-3} \cline{2-3} \cline{3-3} 
$(\Sigma,\Rho)$ & $\rho_0$ & $\rho_1$ & \tabularnewline
\cline{1-3} \cline{2-3} \cline{3-3} 
$\sigma_0$ & $\mu^\sigma_0\mu^\rho_0$ & $\mu^\sigma_0\mu^\rho_1$ & $\mu(\sigma_0)$\tabularnewline
\cline{1-3} \cline{2-3} \cline{3-3} 
$\sigma_1$ & $\mu^\sigma_1\mu^\rho_0$ & $\mu^\sigma_1\mu^\rho_1$ & \tabularnewline
\cline{1-3} \cline{2-3} \cline{3-3} 
\multicolumn{1}{c}{} & \multicolumn{1}{c}{$\mu(\rho_0)$} & \multicolumn{1}{c}{} & \tabularnewline
\end{tabular}
\par\end{centering}
\caption{Misspecified prior $\hat\mu$.}
\label{tab:mprior}
\end{table}

\noindent We call such a misspecified prior $\hat\mu$ a \emph{simplistic worldview}. Lemma \ref{lem:mprior} below demonstrates that  the simplistic worldview relates to the correct prior in a systematic way.

\begin{lemma}\label{lem:mprior} Let $\mu\in\Delta(\Omega)$ and
$k,\ell\in\{0,1\}$. Let \[\hat{\mu}_{kk}-\mu_{kk} \equiv c.\] In this case,
\[\hat{\mu}_{k\ell}-\mu_{k\ell}=-c, \text{ for }  k\neq\ell.\] \end{lemma}

\noindent Lemma \ref{lem:mprior} shows that the amount of change in the priors about the states (due to the simplistic worldview) in which preferences are aligned and misaligned is the same but in opposite directions. 
That is, if the naive receiver perceives the states in which preferences are aligned as more likely than the rational receiver (i.e. if $c>0$), then he perceives the states in which preferences are misaligned as less likely than the rational receiver, by the same amount (and vice versa). 
Table \ref{tab:mpriordiff} demonstrates the misspecified prior $\hat\mu$ in terms of the correct prior $\mu$.  
\begin{table}[H]
\begin{centering}
\begin{tabular}{|c|c|c|c}
\cline{1-3} \cline{2-3} \cline{3-3} 
$(\Sigma,\Rho)$ & $\rho_0$ & $\rho_1$ & \tabularnewline
\cline{1-3} \cline{2-3} \cline{3-3} 
$\sigma_0$ & $\mu_{00}+c$ & $\mu_{01}-c$ & $\mu(\sigma_0)$\tabularnewline
\cline{1-3} \cline{2-3} \cline{3-3} 
$\sigma_1$ & $\mu_{10}-c$ & $\mu_{11}+c$ & \tabularnewline
\cline{1-3} \cline{2-3} \cline{3-3} 
\multicolumn{1}{c}{} & \multicolumn{1}{c}{$\mu(\rho_0)$} & \multicolumn{1}{c}{} & \tabularnewline
\end{tabular}
\par\end{centering}
\caption{$\hat\mu$ expressed in terms of $\mu$.}
\label{tab:mpriordiff}
\end{table}

The naive receiver observes the signal chosen by the sender, who knows that the receiver is naive. 
Moreover, he updates his belief according to the signal chosen by the sender and the misspecified prior $\hat\mu$. 
Let $\hat\mu^m(\rho_\ell)$ denote the posterior belief about state $\rho_\ell$ when the receiver updates according to $\hat\mu$. 
Thus, the receiver's posterior that the state is $\rho_1$ upon observing message $1$ is 
\begin{equation}
\hat\mu^{1}(\rho_1)=\frac{\hat\mu_{01}\pi_{01}+\hat\mu_{11}\pi_{11}}{\hat\mu_{00}\pi_{00}+\hat\mu_{10}\pi_{10}+\hat\mu_{01}\pi_{01}+\hat\mu_{11}\pi_{11}}.\label{upating_BR}
\end{equation}

\noindent While the sender calculates her expected utility according to $\mu$, she knows that the naive receiver updates according to $\hat\mu$. 
Therefore, the sender's optimization problem in (\ref{OP_soph}) becomes 
\begin{align}
\max_{\pi\in\Pi}\: \mu_{11}\pi_{11}+\mu_{10}\pi_{10}&-\mu_{01}\pi_{01}-\mu_{00}\pi_{00}\nonumber\\
\text{s.t. } &\hat\mu_{00}\pi_{00}+\hat\mu_{10}\pi_{10}\leq \hat\mu_{01}\pi_{01}+\hat\mu_{11}\pi_{11}.\label{OP_naive}
\end{align}

Let $\hat\pi^*$ denote the optimal signal when the receiver is naive. 
Note that in (\ref{OP_naive}), similar to (\ref{OP_soph}), the sender's optimal choice in states $(\sigma_0, \rho_0)$ and $(\sigma_1, \rho_1)$ relaxes the obedience constraint and thus $\hat{\pi}^*_{00}=0$ and $\hat{\pi}^*_{11}=1$, which allows us to simplify (\ref{OP_naive}).

\begin{lemma}\label{lemma:naive_R_simpl_OP}
The sender's simplified problem is given by
\begin{equation}
    \begin{aligned}
\max_{\pi\in\Pi} \mu_{10}\pi_{10} & - \mu_{01}\pi_{01}\\
s.t. & \quad\hat{\mu}_{10}\pi_{10}-\hat{\mu}_{01}\pi_{01}\leq\hat{\mu}_{11}.
\label{OP_BR_simplified}
\end{aligned}
\end{equation}
\end{lemma}

We now provide the full characterization of the optimal signal for a naive receiver.

\begin{proposition}\label{pr:signal_corr_neglect} 
An optimal signal when the receiver is naive is given by
\begin{equation*}
(\hat\pi^*_{00},\hat\pi^*_{01},\hat\pi^*_{10},\hat\pi^*_{11})=\begin{cases}
(0,0,1,1) & \text{if }\:\hat\mu_{11}\geq\hat\mu_{10},\\
(0,\beta,\gamma,1) & \text{otherwise},
\end{cases}
\end{equation*}
where
\begin{equation*}
\left(\beta,\gamma\right)=\begin{cases}
\left(0,\frac{\hat{\mu}_{11}}{\hat{\mu}_{10}}\right) & \text{if }\:\frac{\hat{\mu}_{01}}{\hat{\mu}_{10}}<\frac{\mu_{01}}{\mu_{10}},\\
\left(1,\frac{\hat{\mu}_{01}+\hat{\mu}_{11}}{\hat{\mu}_{10}}\right) & \text{if }\:\frac{\hat{\mu}_{01}}{\hat{\mu}_{10}}>\frac{\mu_{01}}{\mu_{10}}\text{ and }\frac{\hat{\mu}_{01}+\hat{\mu}_{11}}{\hat{\mu}_{10}}\leq1,\\
\left(\frac{\hat{\mu}_{10}-\hat{\mu}_{11}}{\hat{\mu}_{01}},1\right) & \text{if }\:\frac{\hat{\mu}_{01}}{\hat{\mu}_{10}}>\frac{\mu_{01}}{\mu_{10}}\text{ and }\frac{\hat{\mu}_{01}+\hat{\mu}_{11}}{\hat{\mu}_{10}}\geq1,
\end{cases}
\vspace{0.25cm}
\end{equation*}
and $\hat{\mu}_{10}\gamma-\hat{\mu}_{01}\beta=\hat{\mu}_{11}$ if $(\hat\mu_{01}/\hat\mu_{10})=(\mu_{01}/\mu_{10})$.
\end{proposition}

The optimal signal in Proposition \ref{pr:signal_corr_neglect} has
several features similar to the one in  Proposition \ref{prop:bench}. 
First, the sender always recommends switching the default action in state $(\sigma_{1},\rho_{1})$ (i.e. $\hat\pi^*_{11}=1$) and never recommends switching in state $(\sigma_{0},\rho_{0})$ (i.e. $\hat{\pi}_{00}^{*}=0$).
Second, when $\hat\mu_{11}\geq\hat\mu_{10}$, the sender can ensure that the receiver takes her desired action with probability 1 (i.e. $\hat{\pi}_{01}^{*}=0$ and $\hat{\pi}_{10}^{*}=1$). 
On the other hand, if $\hat\mu_{11}<\hat\mu_{10}$ then $\hat\pi^*_{01}$ and $\hat\pi^*_{10}$ are pinned down by the relation between the slope of the obedience constraint $(\hat\mu_{01}/\hat\mu_{10})$ and the slope of the objective function $(\mu_{01}/\mu_{10})$.

One very important difference from Proposition \ref{prop:bench}, however, is that when $\mu_{10}\neq\mu_{01}$ (which implies $(\hat\mu_{01}/\hat\mu_{10})\neq(\mu_{01}/\mu_{10})$) the optimal signal in Proposition \ref{pr:signal_corr_neglect} is unique. 
This is a direct consequence of the receiver's naivete, as it leads to the sender and receiver having different priors. 
More precisely, if they have the same prior then the sender is indifferent between marginally increasing $\pi_{10}$ or marginally decreasing $\pi_{01}$.
On the other hand, if they have different priors then this indifference no longer holds as the two options tighten the obedience constraint to a different extent.
Hence, the sender chooses the option that tightens the obedience constraint less as to optimally persuade the receiver.

\section{Welfare implications of a simplistic worldview}
In this section, we proceed beyond characterizing the optimal signal and study a number of general properties of the sender's problem (\ref{OP_naive}) in comparison to the benchmark (\ref{OP_soph}). 
This allows us to characterize the welfare implications of the receiver's naivete and explain the economic intuition behind them. 
Lemma \ref{pr:pre-main} demonstrates how the receiver's naivete
affects the sender's ability to manipulate the receiver's actions.

\begin{lemma}\label{pr:pre-main}
Assume that $c>0$. 
\begin{enumerate}[label=(\roman*)]
    \item For any $\pi\in\Pi$ with $\pi_{00}=0$ and $\pi_{11}=1$ it holds that $\hat{\mu}^{1}(\rho_{1})>\mu^{1}(\rho_{1})$. That is, recommendation 1 induces a higher belief about state $\rho_1$ for a receiver with a simplistic worldview than for a rational receiver. 
    \item For any $\pi\in\Pi$ it holds that $(\hat{\mu}_{11}/\hat{\mu}_{10})+(\hat{\mu}_{01}/\hat{\mu}_{10})\pi_{01}>(\mu_{11}/\mu_{10})+(\mu_{01}/\mu_{10})\pi_{01}$. That is, the obedience constraint in (\ref{OP_BR_simplified}) is more relaxed relative to the benchmark (\ref{OP_soph}). 
\end{enumerate}
If $c<0$, the converse of the inequalities in $(i)$ and $(ii)$ hold. 
\end{lemma}

The intuition of Lemma \ref{pr:pre-main} is as follows. 
First, if $c>0$ then by Lemma \ref{lem:mprior} the naive receiver perceives the aligned states $\left(\sigma_{0},\rho_{0}\right)$ and $\left(\sigma_{1},\rho_{1}\right)$ as more likely (and the two misaligned states as less likely) than the rational receiver. 
Hence, the naive receiver associates recommendation 1 with $\rho_1$ more than the rational receiver and thus has a higher posterior belief.  

Part $(ii)$ of the lemma follows easily from $(i)$. 
Since the sender can induce a posterior belief of at least 1/2 that the state is $\rho_1$ when the receiver is naive by recommending 1 more frequently than before, the obedience constraint in (\ref{OP_BR_simplified}) is more relaxed relative to the benchmark (\ref{OP_soph}).
It is easy to see that the converse of the above arguments hold if $c<0$. 

Importantly, Lemma \ref{pr:pre-main} has direct implications on the sender's and receiver's welfare as it pins down the conditions under which it is ``easier'' for the sender to persuade a naive receiver. 
In particular, if $c>0$ then by $(ii)$ the obedience constraint is more slack than in the benchmark. 
Hence, the sender can decrease $\pi_{01}$ and increase $\pi_{10}$ so that the receiver chooses the sender-preferred actions in the misaligned states $(\sigma_0,\rho_1)$ and $(\sigma_1,\rho_0)$ with a higher probability. 
Clearly, this benefits the sender.
Moreover, since total welfare is the same at the optimum in both cases, it follows that naivete harms the receiver.  
Intuitively, converse arguments hold if $c<0$.

The welfare implications of the receiver's simplistic worldview are summarized in Proposition \ref{pr:main}.
\begin{proposition}\label{pr:main}
Let $c>0$.
A comparison of expected utilities yields the following:  
\begin{enumerate}[label=(\roman*)]
    \item The receiver's simplistic worldview harms him and benefits the sender.
    \item The differences in their expected utilities are strict if and only if $\mu_{11}<\mu_{10}$.
\end{enumerate}
If $c<0$, then the converse of $(i)$ holds and the differences in their expected utilities are strict if and only if $\hat\mu_{11}<\hat\mu_{10}$.
\end{proposition}

Part $(ii)$ of Proposition \ref{pr:main} follows from the structure of the optimal signals. 
If the sender is facing a rational receiver then she can get her preferred outcome with certainty whenever $\mu_{11}\geq\mu_{10}$. 
Hence, in this case if $c>0$ the sender can only weakly benefit from the receiver being naive. 
Otherwise, the sender can always guarantee a strict improvement. 
On the other hand, if $c<0$ and the sender is facing a naive receiver then she can get her preferred outcome with certainty whenever $\hat\mu_{11}\geq\hat\mu_{10}$.
But since the receiver's naivete is detrimental for the sender if $c<0$, it follows that the sender obtains her preferred outcome with certainty also when the receiver is rational.
In this case, the sender is only (weakly) harmed by the receiver's naivete. 

Proposition \ref{pr:main} can also be interpreted in terms of the receiver's misperception of ex ante preference alignment. 
Recall that by Lemma \ref{lem:mprior} if $c>0$ then the naive receiver perceives the aligned states as more likely than the rational receiver. 
Therefore, Lemma \ref{pr:pre-main} implies that it is easier for the sender to persuade the receiver. 
Hence, \emph{naivete makes the receiver more manipulable and is detrimental}. 
Intuitively, if $c<0$ the naivete of the receiver constrains the sender more and thus the ``correct'' recommendation is sent with a higher probability. 
In this case, \emph{naivete makes the receiver less manipulable and is beneficial}.

\section{Conclusion}
Real-world examples of Bayesian persuasion problems, such as corporate consulting, often involve the sender and receiver caring about different, but connected issues of a problem. 
This paper studies the impact of the receiver's naivete regarding the interconnected nature of the two dimensions of the state of the world on the optimal outcomes for both parties. 
First, as a benchmark, we study the case of a rational receiver, i.e., the sender and receiver both know the correct prior distribution of the states. 
We show that regardless of the receiver being rational or naive, it is optimal for the sender to truthfully disclose the states whenever her preference is aligned with the receiver. 
On the other hand, optimal disclosure in states where preferences are misaligned crucially depends on the receiver's perception of the prior.

Our findings reveal that the receiver's naivete can disadvantage him while favoring the sender, particularly when it leads the receiver to overestimate the likelihood of states where both parties prefer the same action and underestimate states where they have conflicting preferences.
Intuitively, this makes the receiver ``easier'' to persuade; the sender can commit to sending the ``wrong'' message with a higher probability. 
Conversely, if the naivete leads the receiver to perceive aligned states as less likely and misaligned states as more probable, the receiver benefits at the expense of the sender.

\appendixtitleon
\appendixtitletocon
\begin{appendices}

\section{Proofs}
\setcounter{lemma}{0}
    \renewcommand{\thelemma}{\Alph{section}\arabic{lemma}}

\noindent\textbf{\emph{Proof of Lemma \ref{lem:str}}.}
For the obedience constraint to be satisfied, it must hold that $\mu^{1}(\rho_1)\geq 1/2$, which is equivalent to 
\begin{equation*}
\frac{\mu_{01}\pi_{01}+\mu_{11}\pi_{11}}{\mu_{01}\pi_{01}+\mu_{11}\pi_{11}+\mu_{10}\pi_{10}+\mu_{00}\pi_{00}}\geq\frac12.
\end{equation*}

\noindent Rearranging the inequality yields the obedience constraint.\qed

\begin{lemma}\label{lem:suff}
Let $\pi\in\Pi$. 
If $\mu_{00}\pi_{00}+\mu_{10}\pi_{10}\leq \mu_{01}\pi_{01}+\mu_{11}\pi_{11}$, then $\mu_{01}(1-\pi_{01})+\mu_{11}(1-\pi_{11})<\mu_{00}(1-\pi_{00})+\mu_{10}(1-\pi_{10})$.
\end{lemma}
\begin{proof}
Let $\mu_{00}\pi_{00}+\mu_{10}\pi_{10}\leq \mu_{01}\pi_{01}+\mu_{11}\pi_{11}$, which implies $-\mu_{01}\pi_{01}-\mu_{11}\pi_{11}\leq -\mu_{00}\pi_{00}-\mu_{10}\pi_{10}$.
Recall that $\mu(\rho_1)<\mu(\rho_0)$, i.e. $\mu_{01}+\mu_{11}<\mu_{00}+\mu_{10}$.
Hence, it follows that $$\mu_{01}+\mu_{11}-\mu_{01}\pi_{01}-\mu_{11}\pi_{11}<\mu_{00}+\mu_{10}-\mu_{00}\pi_{00}-\mu_{10}\pi_{10}.$$
Rewriting the above inequality provides the desired result. 
Hence, the obedience constraint for $0$ is satisfied.
That is, 
\begin{equation*}
    \mu^{0}(\rho_0)>\frac{\mu_{10}(1-\pi_{10})+\mu_{00}(1-\pi_{00})}{\mu_{10}(1-\pi_{10})+\mu_{00}(1-\pi_{00})+\mu_{01}(1-\pi_{01})+\mu_{11}(1-\pi_{11})}>\frac12.
\end{equation*}
\end{proof}

\noindent\textbf{\emph{Proof of Proposition \ref{prop:bench}}.}
First suppose that $\mu_{10}\leq\mu_{11}$.
It is clear that setting $\pi^*=(\pi^*_{00},\pi^*_{01},\pi^*_{10},\pi^*_{11})=(0,0,1,1)$ satisfies the obedience constraint in (\ref{OP_soph}) and maximizes the sender's expected utility since it provides the upper bound of 1. 
Hence, $\pi^*$ is the unique solution to the sender's optimization problem. 

Now suppose that $\mu_{10}>\mu_{11}$. 
First observe that it is optimal in this case to set $\pi^*_{11}=1$ and $\pi^*_{00}=0$ as well, since it increases the sender's expected utility and relaxes the obedience constraint.  
Hence, substituting $\pi^*_{11}=1$ and $\pi^*_{00}=0$ into the sender's optimization problem and simplifying we obtain
\begin{align}
\max_{\pi\in\Pi}\:\mu_{10}\pi_{10}&-\mu_{01}\pi_{01}\label{opt:sop2}\\ 
\text{s.t. } &\mu_{10}\pi_{10}\leq \mu_{01}\pi_{01}+\mu_{11}.\nonumber
\end{align}

\noindent It is easy to see from (\ref{opt:sop2}) that the objective function is increasing in $\pi_{10}$ and decreasing in $\pi_{01}$.
Hence, if (\ref{opt:sop2}) is not binding then the sender can marginally increase $\pi_{10}$ or decrease $\pi_{01}$ to achieve a higher expected payoff. 
Thus, the constraint is binding at the optimum. 

Therefore, since both the objective function and the constraint are linear (in $\pi_{10}$ and $\pi_{01}$), there exists a continuum of solutions to (\ref{opt:sop2}).\qed\\

\noindent\textbf{\emph{Proof of Lemma \ref{lem:mprior}}.}
Let $\hat\mu_{kk}-\mu_{kk}=c$ and assume that $k\neq\ell$. 
Then,  
\begin{align*}
    \hat\mu_{k\ell}-\mu_{k\ell}=\mu(\sigma_k)(1-\mu(\rho_k))-\mu_{k\ell}&=\mu_{k\ell}+\mu_{kk}-\mu(\sigma_k)\mu(\rho_k)-\mu_{k\ell}\\
    &=\mu(\sigma_k)\mu(\rho_k)-c-\mu(\sigma_k)\mu(\rho_k)=-c.
\end{align*}
\qed

\noindent\textbf{\emph{Proof of Lemma \ref{lemma:naive_R_simpl_OP}}.}
As one can observe from (\ref{OP_naive}), the sender's payoff is increasing in $\pi_{11}$, which relaxes the obedience constraint.  Thus, $\hat{\pi}_{11}^* = 1$. Similar reasoning leads to $\hat{\pi}_{00}^* = 0$. Substituting $\hat{\pi}_{00}^{*}=0,\hat{\pi}_{11}^{*}=1$ into (\ref{OP_naive}) yields the result. \qed\\

\noindent\textbf{\emph{Proof of Proposition \ref{pr:signal_corr_neglect}}.}
The sender's problem is given by
\begin{align*}
\max_{\pi\in\Pi} \mu_{10}\pi_{10} & - \mu_{01}\pi_{01}\\
s.t. & \quad\hat{\mu}_{10}\pi_{10}-\hat{\mu}_{01}\pi_{01}\leq\hat{\mu}_{11}.
\end{align*} 
As the problem is a linear program, the solution lies either on the obedience constraint or on one of the boundaries of the set $\Pi$. 
First, suppose that $\hat{\mu}_{11}\geq \hat{\mu}_{10}$.
Then the solution lies on the boundary and the obedience constraint is inactive for $\hat{\mu}_{11} > \hat{\mu}_{10}$.
It follows easily that $\pi_{01}^{*} = 0,\pi_{10}^{*} = 1$. 

Next, suppose that $\hat{\mu}_{11}< \hat{\mu}_{10}$. 
All three possible types of solutions for this case are presented in Figure \ref{fig:pr_corr_negl_signal_proof}. 
We proceed with characterizing each of the three cases.

\begin{figure}
    \centering
    \includegraphics[width = 0.9\textwidth]{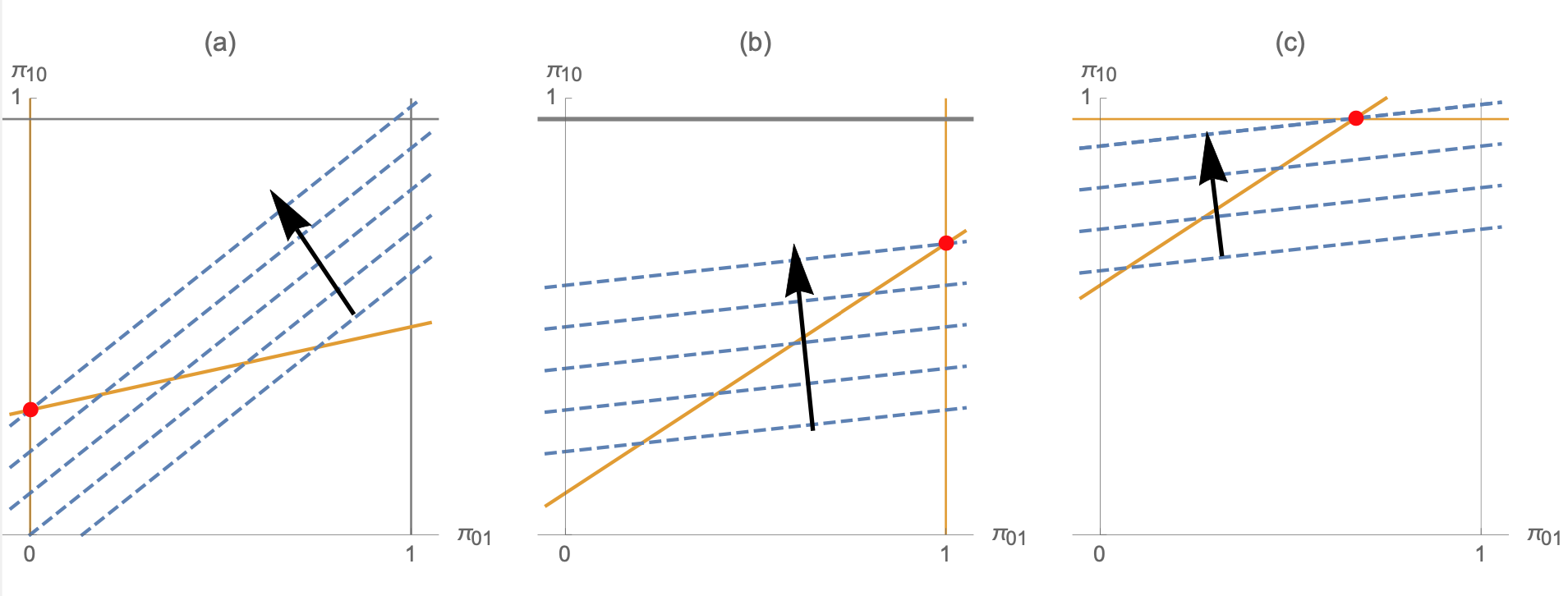}
\setlength{\abovecaptionskip}{0pt}
    \caption{Three possible forms of solutions to the problem (\ref{OP_BR_simplified}). Indifference curves are in blue and dashed. Active constraints are in orange.}
    \label{fig:pr_corr_negl_signal_proof}
\end{figure}

\textbf{Case (a)} The active constraints are $\pi_{01}\geq0$ and $\hat{\mu}_{10}\pi_{10}-\hat{\mu}_{01}\pi_{01}\leq\hat{\mu}_{11}$.
Hence, $\hat{\pi}_{01}^{*}=0$ and $\hat{\pi}_{10}^{*}$
solves $\hat{\mu}_{10}\pi_{10}-\hat{\mu}_{01}\pi_{01}=\hat{\mu}_{11}$, i.e. $\pi^*_{10}=\hat\mu_{11}/\hat\mu_{10}$.
The solution to (\ref{OP_BR_simplified}) has this particular form
if and only if the slope of the objective function is steeper than
the slope of the obedience constraint, i.e. $(\hat{\mu}_{01}/\hat{\mu}_{10})<(\mu_{01}/\mu_{10})$.

\textbf{Case (b)} The active constraints are $\pi_{01}\leq1$ and $\hat{\mu}_{10}\pi_{10}-\hat{\mu}_{01}\pi_{01}\leq\hat{\mu}_{11}$.
Hence, $\hat{\pi}_{01}^{*}=1$ and $\hat{\pi}_{10}^{*}$
solves $\hat{\mu}_{10}\pi_{10}-\hat{\mu}_{01}\pi_{01}=\hat{\mu}_{11}$, i.e. $\hat{\pi}_{10}^{*}=(\hat{\mu}_{01}+\hat{\mu}_{11})/\hat{\mu}_{10}$.
The solution to (\ref{OP_BR_simplified}) has this form if and only
if (i) $(\hat{\mu}_{01}/(\hat{\mu}_{10})>(\mu_{01}/\mu_{10})$ (i.e. the slope of the objective function is flatter than the slope of the obedience constraint)  
and (ii)  $(\hat{\mu}_{01}+\hat{\mu}_{11})/\hat{\mu}_{10}\leq 1$ (which implies that the constraint $\pi_{10}\leq1$ is inactive). 

\textbf{Case (c)} The active constraints are $\pi_{10}\leq1$ and $\hat{\mu}_{10}\pi_{10}-\hat{\mu}_{01}\pi_{01}\leq\hat{\mu}_{11}$.
Hence, $\hat{\pi}_{10}^{*}=1$ and $\hat{\pi}_{01}^{*}$
solves $\hat{\mu}_{10}\pi_{10}-\hat{\mu}_{01}\pi_{01}=\hat{\mu}_{11}$, i.e. $\hat{\pi}_{01}^{*}=(\hat{\mu}_{10}-\hat{\mu}_{11})/\hat{\mu}_{01}$.
The solution to (\ref{OP_BR_simplified}) has this form if and only
if (i) $\hat{\mu}_{01}/\hat{\mu}_{10}>\mu_{01}/\mu_{10}$,
and (ii)  $(\hat{\mu}_{01}+\hat{\mu}_{11})/\hat{\mu}_{10}\geq 1$ (which implies that the constraint $\pi_{10}\leq1$ is active).

Finally, if $(\hat{\mu}_{01}/\hat{\mu}_{10})=(\mu_{01}/\mu_{10})$, then any $\pi_{10}$ and $\pi_{01}$ such that $\hat{\mu}_{10}\pi_{10}-\hat{\mu}_{01}\pi_{01}=\hat{\mu}_{11}$
solves (\ref{OP_BR_simplified}). \qed\\

\noindent\textbf{\emph{Proof of Lemma \ref{pr:pre-main}}.}
We begin by proving part $(i)$. 
Since it is optimal to set $\hat\pi^*_{00}=0$ and $\hat\pi^*_{11}=1$, (\ref{upating_BR}) can be rewritten as
\[
\hat{\mu}^{1}\left(\rho_{1}\right)=\frac{\hat{\mu}_{01}\pi_{01}+\hat{\mu}_{11}}{\hat{\mu}_{01}\pi_{01}+\hat{\mu}_{11}+\hat{\mu}_{10}\pi_{10}}.
\]
By Lemma \ref{lem:mprior}, it follows that
\[
\begin{aligned}
\hat{\mu}^{1}\left(\rho_{1}\right) &=\frac{\left(\mu_{01}-c\right)\pi_{01}+\left(\mu_{11}+c\right)}{\left(\mu_{01}-c\right)\pi_{01}+\left(\mu_{11}+c\right)+\left(\mu_{10}-c\right)\pi_{10}}\\
&=\frac{\mu_{01}\pi_{01}+\mu_{11}+c\left(1-\pi_{01}\right)}{\mu_{01}\pi_{01}+\mu_{11}+\mu_{10}\pi_{10}+c\left(1-\pi_{01}-\pi_{10}\right)}.
\end{aligned}
\]
Suppose that $c>0$.
Then it holds that $c\left(1-\pi_{01}\right)>c\left(1-\pi_{01}-\pi_{10}\right)$,
and thus 
\[
\mu^{1}\left(\rho_{1}\right)=\frac{\mu_{01}\pi_{01}+\mu_{11}}{\mu_{01}\pi_{01}+\mu_{11}+\mu_{10}\pi_{10}}<\frac{\mu_{01}\pi_{01}+\mu_{11}+c\left(1-\pi_{01}\right)}{\mu_{01}\pi_{01}+\mu_{11}+\mu_{10}\pi_{10}+c\left(1-\pi_{01}-\pi_{10}\right)}=\hat{\mu}^{1}\left(\rho_{1}\right).
\]

Next suppose that $c<0$.
Then it holds that $c\left(1-\pi_{01}\right)<c\left(1-\pi_{01}-\pi_{10}\right)$,
and thus
\[
\mu^{1}\left(\rho_{1}\right)=\frac{\mu_{01}\pi_{01}+\mu_{11}}{\mu_{01}\pi_{01}+\mu_{11}+\mu_{10}\pi_{10}}>\frac{\mu_{01}\pi_{01}+\mu_{11}+c\left(1-\pi_{01}\right)}{\mu_{01}\pi_{01}+\mu_{11}+\mu_{10}\pi_{10}+c\left(1-\pi_{01}-\pi_{10}\right)}=\hat{\mu}^{1}\left(\rho_{1}\right).
\]

Now we prove part $(ii)$. 
Recall the sender's simplified problems (\ref{OP_soph}) for the rational receiver and (\ref{OP_BR_simplified}) for the naive receiver. 
Rewriting the constraints in (\ref{OP_soph}) and (\ref{OP_BR_simplified}), respectively, yields
\[
\pi_{10}\leq\frac{\mu_{01}}{\mu_{10}}\pi_{01}+\frac{\mu_{11}}{\mu_{10}}, \hspace{10pt} \pi_{10}\leq\frac{\hat{\mu}_{01}}{\hat{\mu}_{10}}\pi_{01}+\frac{\hat{\mu}_{11}}{\hat{\mu}_{10}}.
\]

Define $g_{1}\left(\pi_{01}\right)\equiv(\mu_{01}/\mu_{10})\pi_{01}+(\mu_{11}/\mu_{10})$ and $g_{2}\left(\pi_{01}\right)\equiv(\hat{\mu}_{01}/\hat{\mu}_{10})\pi_{01}+(\hat{\mu}_{11}/\hat{\mu}_{10})$.
By Lemma \ref{lem:mprior} it follows that
\[
g_{2}\left(\pi_{01}\right)=\frac{\mu_{01}-c}{\mu_{10}-c}\pi_{01}+\frac{\mu_{11}+c}{\mu_{10}-c}.
\]
Consider $g_{1}\left(\cdot\right)$ and $g_{2}\left(\cdot\right)$ at the
boundaries of the set of feasible signals: $g_{1}\left(0\right)=(\mu_{11}/\mu_{10}),$
$g_{2}\left(0\right)=(\mu_{11}+c)/(\mu_{10}-c)$ and $g_{1}\left(1\right)=(\mu_{01}+\mu_{11})/\mu_{10},$
$g_{2}\left(1\right)=(\mu_{01}+\mu_{11})/(\mu_{10}-c)$.

First suppose that $c>0$. 
It follows that $g_{1}\left(0\right)<g_{2}\left(0\right)$
and $g_{1}\left(1\right)<g_{2}\left(1\right)$. 
Moreover, by linearity
of $g_{1}\left(\cdot\right)$ and $g_{2}\left(\cdot\right)$, for all $\pi_{01}\in[0,1]$ it holds that
\[
g_{1}\left(\pi_{01}\right)<g_{2}\left(\pi_{01}\right).
\]
Thus, the constraint in (\ref{OP_soph}) is relatively more slack than in (\ref{OP_BR_simplified}).

Next suppose that $c<0$. 
It follows that $g_{1}\left(0\right)>g_{2}\left(0\right)$
and $g_{1}\left(1\right)>g_{2}\left(1\right)$. 
Again, by linearity
of $g_{1}\left(\cdot\right)$ and $g_{2}\left(\cdot\right)$, for all $\pi_{01}\in[0,1]$ it holds that
\[
g_{1}\left(\pi_{01}\right)>g_{2}\left(\pi_{01}\right).
\]
Thus, the constraint in (\ref{OP_soph}) is relatively more tight than in (\ref{OP_BR_simplified}).
\qed \\

\noindent\textbf{\emph{Proof of Proposition \ref{pr:main}}.} 
Let $V^*$ and $\hat V^*$ denote the sender's optimal (ex-ante) expected payoff when the receiver is rational and when he is naive, respectively. 
Then,
\[
\begin{aligned}
V^{*}\equiv\mu_{11}\pi_{11}^{*}+\mu_{10}\pi_{10}^{*}+\mu_{01}(1-\pi_{01}^{*})+\mu_{00}(1-\pi_{00}^{*}),\\
\hat{V}^{*}\equiv\mu_{11}\hat{\pi}_{11}^{*}+\mu_{10}\hat{\pi}_{10}^{*}+\mu_{01}(1-\hat{\pi}_{01}^{*})+\mu_{00}(1-\hat{\pi}_{00}^{*}).
\end{aligned}\]

As the form of optimal signal in Propositions \ref{prop:bench} and \ref{pr:signal_corr_neglect} depends on $sign(\mu_{11}- \mu_{10})$, we consider four cases based on $sign(\mu_{11}- \mu_{10})$ and $sign(c)$.
First, assume $\mu_{11}\geq\mu_{10}$. If $c\geq0$, then from Proposition \ref{prop:bench} and \ref{pr:signal_corr_neglect} and Lemma \ref{lem:mprior},  $(\pi_{00}^{*},\pi_{01}^{*},\pi_{10}^{*},\pi_{11}^{*})=$ $(\hat{\pi}_{00}^{*},\hat{\pi}_{01}^{*},\hat{\pi}_{10}^{*},\hat{\pi}_{11}^{*})=(0,0,1,1)$. 
Thus, $V^{*}=\hat{V}^{*}=1$. If $c<0$, then either $\hat{\mu}_{11}\geq\hat{\mu}_{10}$ and, as above, $V^{*}=\hat{V}^{*}=1$, or $\hat{\mu}_{11}<\hat{\mu}_{10}$ and (as, according to Proposition \ref{pr:signal_corr_neglect}, the sender does not get her first best payoff) $\hat{V}^{*}<V^{*}=1$. Thus, it follows that 
\begin{equation}
\hat{V}^{*}\geq V^{*} \text{ if and only if } c\geq0.
\label{eq:V_compar_result_1}
\end{equation}

Second, assume $\mu_{11}<\mu_{10}$.
By Proposition \ref{prop:bench}, the difference in the sender's optimal expected payoffs between
the benchmark and the naive receiver case is given by
\begin{equation}
\begin{aligned}\nu\equiv\hat{V}^{*}-V^{*} & =\mu_{10}\left(\hat{\pi}_{10}^{*}-\pi_{10}^{*}\right)+\mu_{01}\left(\pi_{01}^{*}-\hat{\pi}_{01}^{*}\right)\\
 & =\mu_{10}\left(\hat{\pi}_{10}^{*}-\gamma\right)+\mu_{01}\left(\frac{\mu_{10}}{\mu_{01}}\gamma-\frac{\mu_{11}}{\mu_{01}}-\hat{\pi}_{01}^{*}\right)\\
 & =\mu_{10}\hat{\pi}_{10}^{*}-\mu_{01}\hat{\pi}_{01}^{*}-\mu_{11}.
\end{aligned} \label{eq:V_compar}
\end{equation}

Plugging the optimal signal for case (a) from Proposition  \ref{pr:signal_corr_neglect} in (\ref{eq:V_compar}) yields
\begin{equation}
\begin{aligned}
\nu &=\mu_{10}\frac{\hat{\mu}_{11}}{\hat{\mu}_{10}}-\mu_{11}
=\mu_{10}\frac{\mu_{11}+c}{\hat{\mu}_{10}}-\mu_{11} \\&=\frac{\mu_{10}\mu_{11}+\mu_{11}c-\mu_{10}\mu_{11}+\mu_{10}c}{\hat{\mu}_{10}}=\frac{c\left(\mu_{11}+\mu_{10}\right)}{\hat{\mu}_{10}}.
\end{aligned} \label{eq:main_a}
\end{equation}

Similarly, plugging the optimal signal for case (b) from Proposition  \ref{pr:signal_corr_neglect} in (\ref{eq:V_compar}) yields
\begin{equation}
\begin{aligned}
\nu=&\mu_{10}\frac{\hat{\mu}_{01}+\hat{\mu}_{11}}{\hat{\mu}_{10}}-\mu_{01}-\mu_{11}\\
&=\mu_{10}\frac{\mu_{01}-c+\mu_{11}+c}{\hat{\mu}_{10}}-\mu_{01}-\mu_{11}=\frac{c\left(\mu_{11}+\mu_{01}\right)}{\hat{\mu}_{10}}.
\end{aligned} \label{eq:main_b}
\end{equation}

Finally, plugging  the optimal signal for case (c) from Proposition  \ref{pr:signal_corr_neglect} in (\ref{eq:V_compar}) yields
\begin{equation}
\begin{aligned}
\nu&=\mu_{10}-\mu_{01}\frac{\hat{\mu}_{10}-\hat{\mu}_{11}}{\hat{\mu}_{01}}-\mu_{11}\\&=
\frac{2c\mu_{01}+\left(\mu_{10}-\mu_{11}\right)\left(\hat{\mu}_{01}\right)-\mu_{01}\left(\mu_{10}-\mu_{11}\right)}{\hat{\mu}_{01}}=\frac{c\left(2\mu_{01}+\mu_{11}-\mu_{10}\right)}{\hat{\mu}_{01}}\\&=
\frac{c\left(\mu(\rho_1)-\mu_{10}+\mu_{01}\right)}{\hat{\mu}_{01}}.
\end{aligned} \label{eq:main_c} 
\end{equation}
From the parametric condition for the case (c), $(\hat{\mu}_{01}+\hat{\mu}_{11})/\hat{\mu}_{10}\geq 1$ and thus it follows that $\mu(\rho_1)\geq\hat{\mu}_{10}$.
By Lemma \ref{lem:mprior}, this can be rewritten as $\mu(\rho_1) - \mu_{10}\geq  - c$. 
Since we also have $\mu_{01}\geq c$ by Lemma \ref{lem:mprior}, it holds that
\begin{equation}
    \mu(\rho_1)-\mu_{10}+\mu_{01}\geq0.
    \label{eq:case_c_cond}
\end{equation}

By (\ref{eq:main_a}), (\ref{eq:main_b}),  (\ref{eq:main_c}), and (\ref{eq:case_c_cond}) it follows that,
\begin{equation}
\hat{V}^{*}\geq V^{*} \text{ if and only if } c\geq0.
\label{eq:V_compar_result}
\end{equation}

Note that the sum of optimal expected payoffs of the sender and the receiver in the benchmark and  the case of naive receiver is given by the same constant: 
\begin{equation}
\begin{aligned}[b]
&\underbrace{\mu_{11}\pi_{11}^{*} + \mu_{00}(1-\pi_{00}^{*}) + \mu_{10}\pi_{10}^{*} + \mu_{01}(1-\pi_{01}^{*})}_{\text{the sender's expected payoff}} \\
&\quad\quad + \underbrace{\mu_{11}\pi_{11}^{*} + \mu_{00}(1-\pi_{00}^{*}) + \mu_{10}\left(1-\pi_{10}^{*}\right) + \mu_{01}\pi_{01}^{*}}_\text{the receiver's expected payoff} \\
&= 2\left(\mu_{11}\pi_{11}^{*} + \mu_{00}(1-\pi_{00}^{*})\right) + \mu_{10} + \mu_{01} \\
&= 2\left(\mu_{11} + \mu_{00}\right) + \mu_{10} + \mu_{01},
\end{aligned} \label{eq:sum_of_S_and_R_util}
\end{equation}
where both use the correct (sophisticated) prior belief for evaluating the receiver's expected utility when he is naive.   

Denote the receiver's sender-optimal expected payoff in the benchmark and the case of naive receiver by $U^*$ and $\hat{U}^{*}$, respectively. 
Since by (\ref{eq:sum_of_S_and_R_util}) the sum of sender's and receiver's expected payoffs does not change between the benchmark and the case of naive receiver and the sender's expected payoff changes according to (\ref{eq:V_compar_result_1}) and (\ref{eq:V_compar_result}), it follows that
\begin{equation}
\hat{U}^{*}\leq U^{*} \text{ if and only if } c\geq0,
\label{eq:V_compar_result_2}
\end{equation}
i.e. when the receiver's naivete yields a higher expected payoff to  the sender, it yields a lower expected payoff to the receiver.  

Finally, given $c>0$, from Lemma \ref{pr:pre-main}, the receiver's naivete relaxes the obedience constraint. 
Thus, the difference between the expected payoffs in (\ref{eq:V_compar_result_1}), (\ref{eq:V_compar_result}), and (\ref{eq:V_compar_result_2}) is strict if and only if in the benchmark case the sender does not obtain her first best payoff of $1$, which is the case if and only if $\mu_{11}<\mu_{10}$. 
Similarly, given $c<0$, by Lemma \ref{pr:pre-main}, the receiver's naivete tightens the obedience constraint. 
Thus, the difference between the expected payoffs in (\ref{eq:V_compar_result_1}), (\ref{eq:V_compar_result}), and (\ref{eq:V_compar_result_2}) is strict if and only if the sender does not obtain her first best payoff of $1$ in the naive receiver case, which is true if and only if $\hat{\mu}_{11}<\hat{\mu}_{10}$.\qed \\

\end{appendices}

\addcontentsline{toc}{section}{References} 
\bibliography{main}
\bibliographystyle{chicago}

\end{document}